\documentclass[12pt,reqno]{article}

\usepackage[usenames]{color}
\usepackage{amssymb}
\usepackage{amsmath}
\usepackage{amsthm}
\usepackage{amsfonts}
\usepackage{amscd}
\usepackage{graphicx}

\usepackage[colorlinks=true,
linkcolor=webgreen,
filecolor=webbrown,
citecolor=webgreen]{hyperref}

\definecolor{webgreen}{rgb}{0,.5,0}
\definecolor{webbrown}{rgb}{.6,0,0}

\usepackage{color}
\usepackage{fullpage}
\usepackage{float}

\usepackage{graphics}
\usepackage{latexsym}
\usepackage{epsf}
\usepackage{breakurl}

\newcommand{\seqnum}[1]{\href{https://oeis.org/#1}{\rm \underline{#1}}}

\newcommand{\Enn}{\mathbb{N}}

\begin{document}

\theoremstyle{plain}
\newtheorem{theorem}{Theorem}
\newtheorem{corollary}[theorem]{Corollary}
\newtheorem{lemma}[theorem]{Lemma}
\newtheorem{proposition}[theorem]{Proposition}

\theoremstyle{definition}
\newtheorem{definition}[theorem]{Definition}
\newtheorem{example}[theorem]{Example}
\newtheorem{conjecture}[theorem]{Conjecture}

\theoremstyle{remark}
\newtheorem{remark}[theorem]{Remark}

\def\modd#1 #2{#1\ \mbox{\rm (mod}\ #2\mbox{\rm )}}

\title{Proving properties of some greedily-defined integer recurrences
via automata theory}

\author{Jeffrey Shallit\thanks{Research supported by NSERC, grant 2018-04118.}\\
School of Computer Science\\
University of Waterloo\\
Waterloo, ON  N2L 3G1 \\
Canada\\
\href{mailto:shallit@uwaterloo.ca}{\tt shallit@uwaterloo.ca}}

\maketitle

\begin{abstract}
Venkatachala on the one hand, and
Avdispahi\'c \& Zejnulahi on the other, both studied integer sequences
with an unusual sum property defined in a greedy way, and proved
many results about them.  However, their proofs were rather lengthy
and required numerous cases.  In this paper, I provide a different
approach, via finite automata, that can prove the same results (and more)
in a simple, unified way.  Instead of case analysis, we use
a decision procedure implemented in the free software
{\tt Walnut}.   Using these ideas, we can prove a conjecture of Quet
and find connections between Quet's sequence and the ``married'' functions
of Hofstadter.
\end{abstract}

\section{Introduction}

Let $\Enn = \{0,1,2,\ldots\}$ denote the natural numbers.
In 2009, B. J. Venkatachala 
\cite{Venkatachala:2009} studied the properties of
an amazing sequence of natural numbers
$(f(n))_{n \geq 0}$, whose first
few values are given in Table~\ref{tab1}.

The sequence $f(n)$ can be defined inductively, using a greedy
algorithm, as follows:
$f(0) = 0$, and for $n\geq 1$,
$f(n)$ is the least natural number such that
\begin{itemize}
\item[(a)]
$f(n) \not\in \{f(0), f(1), \ldots, f(n-1) \}$;
\item[(b)]
the sum $\sum_{1 \leq i \leq n} f(i)$ is divisible by $n$.
\end{itemize}
The problem of constructing such a sequence was
proposed earlier by Shapovalov \cite{Shapovalov:1996}.

The related sequence $h$ is defined by the equation 
$$ h(n) = {1 \over n} \bigl(f(1) + \cdots + f(n) \bigr) $$
for $n \geq 1$.  Table~\ref{tab1} gives the first few values.

\begin{table}[H]
\begin{center}
\begin{tabular}{c|ccccccccccccccccccccc}
$n$ & 0 & 1 & 2 & 3 & 4 & 5 & 6 & 7 & 8 & 9 & 10 & 11 & 12 & 13 & 14 & 15 & 16 & 17 & 18 & 19  \\
\hline
$f(n)$ & 0& 1& 3& 2& 6& 8& 4&11& 5&14&16& 7&19&21& 9&24&10&27&29&12\\
$h(n)$ & 0& 1& 2& 2& 3& 4& 4& 5& 5& 6& 7& 7& 8& 9& 9&10&10&11&12&12
\end{tabular}
\end{center}
\caption{Table of the first few values of the sequences $f$ and $h$.}
\label{tab1}
\end{table}
The sequence $f$ is sequence \seqnum{A019444} in the 
OEIS \cite{oeis}, and the sequence
$h$ is sequence \seqnum{A019446} in the OEIS.
Also note that $f(n) = \seqnum{A002251}(n-1) + 1$ for $n \geq 1$.

Venkatachala obtained many interesting results about these numbers,
but his proofs required long case analysis.
In this note we show how to obtain very simple proofs
of these results, and many others
in Venkatachala's paper, with {\tt Walnut}, a theorem-prover
for automatic sequences \cite{Mousavi:2016,Shallit:2022}.
We can also obtain some new results.
It simply suffices to state the theorems as first-order
logic assertions, and let {\tt Walnut} verify them.  All verifications
in this paper were done in a matter of seconds on a laptop.

Similarly, in 2020, Avdispahi\'c \& Zejnulahi
\cite{Avdispahic&Zejnulahi:2020} studied two sequences
$(z(n))_{n \geq 0}$ and $(m(n))_{n \geq 0}$ tabulated in 
Table~\ref{tab2}.

The sequence $z(n)$ can be defined inductively, using
a greedy algorithm, as follows:
$z(0) = 0$, and for $n\geq 1$,
$z(n)$ is the least natural number such that
\begin{itemize}
\item[(a)]
$z(n) \not\in \{z(0), z(1), \ldots, z(n-1) \}$;
\item[(b)]
the sum $\sum_{2 \leq i \leq n} z(i)$ is divisible by $n+1$.
\end{itemize}
The sequence $m$ is defined by the relation
$$ m(n) = {1 \over {n+1}} \bigl(z(2) + \cdots + z(n) \bigr) $$
for $n \geq 1$.  Table~\ref{tab2} gives the first few values.

\begin{table}[H]
\begin{center}
\begin{tabular}{c|ccccccccccccccccccccc}
$n$ & 0 & 1 & 2 & 3 & 4 & 5 & 6 & 7 & 8 & 9 & 10 & 11 & 12 & 13 & 14 & 15 & 16 & 17 & 18 & 19  \\ 
\hline
$z(n)$ & 0& 1& 3& 5& 2& 8&10& 4&13&15& 6&18& 7&21&23& 9&26&28&11&31\\
$m(n)$ & 0& 0& 1& 2& 2& 3& 4& 4& 5& 6& 6& 7& 7& 8& 9& 9&10&11&11&12
\end{tabular}
\end{center}
\caption{Table of the first few values of the sequences $z$ and $m$.}
\label{tab2} 
\end{table}
The sequence $z(n)$ is sequence \seqnum{A340510} in the OEIS,
and the sequence $m(n)$ is sequence \seqnum{A005379}.
This latter sequence also appears in
Hofstadter's celebrated book
\cite[p.~137]{Hofstadter:1979} and was analyzed by Stoll
\cite{Stoll:2008}.

Avdispahi\'c \& Zejnulahi obtained many results on these sequences.  Once again, we can
obtain the same results in a much simpler way using finite
automata.

\section{Automata for the Venkatachala sequences}
\label{sec2}

As Venkatachala observed, these sequences are related to the
golden ratio $\varphi = (1+\sqrt{5})/2$.  So we might suspect
they are related to the so-called Zeckendorf
numeration system \cite{Zeckendorf:1972}; in fact, that they
are computable by a finite automaton.

In this numeration system, natural numbers are represented
as binary strings; the string $a_1 a_2 \cdots a_t$
represents the number $\sum_{1 \leq i \leq t} a_i F_{t+2-i}$,
where $F_0 = 0$, $F_1 = 1$, and $F_n = F_{n-1} + F_{n-2}$ are
the Fibonacci numbers.
In general, numbers may have multiple representations, but
uniqueness is obtained if we insist that $a_i a_{i+1} = 0$
for all $i$.

We can use the ``guessing'' procedure discussed in 
\cite{Shallit:2022} to find candidate automata
computing the functions $f$ and $h$,
based on actual computed initial values. They are
depicted in Figures~\ref{fig1} and \ref{fig2}.  We call
these automata {\tt fp} and {\tt hp}, respectively,
and the functions they compute $f'$ and $h'$.
\begin{figure}[htb]
\begin{center}
\includegraphics[width=6in]{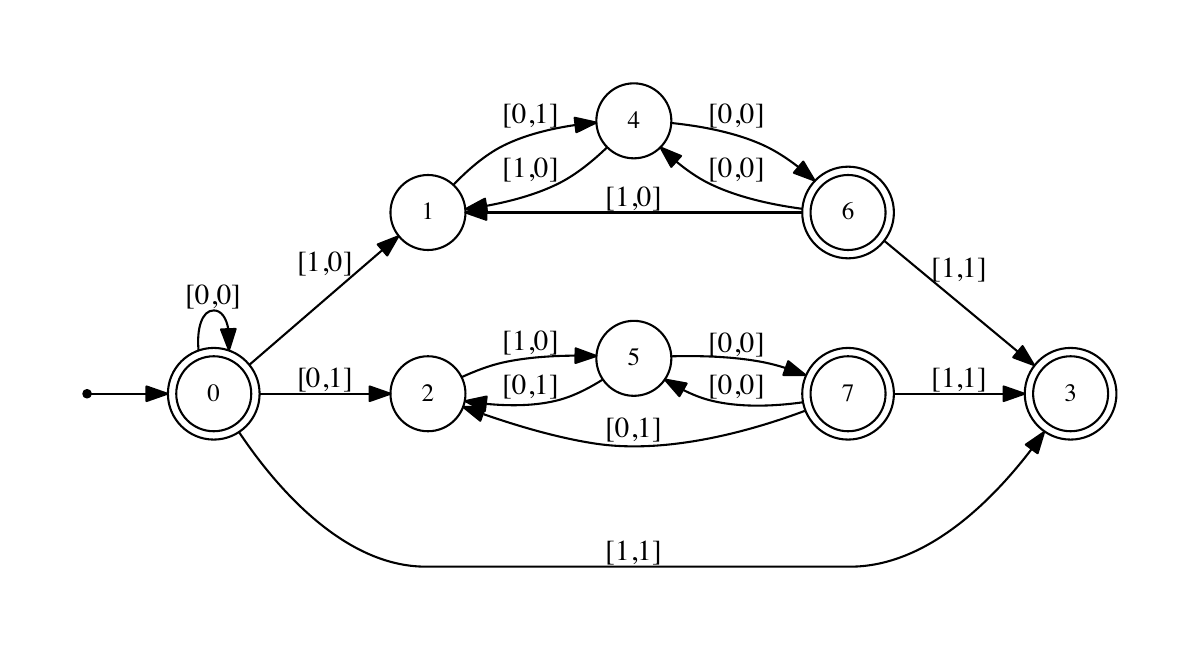}
\end{center}
\caption{Fibonacci automaton for $f'(n)$.}
\label{fig1}
\end{figure}

\begin{figure}[htb]
\begin{center}
\includegraphics[width=6in]{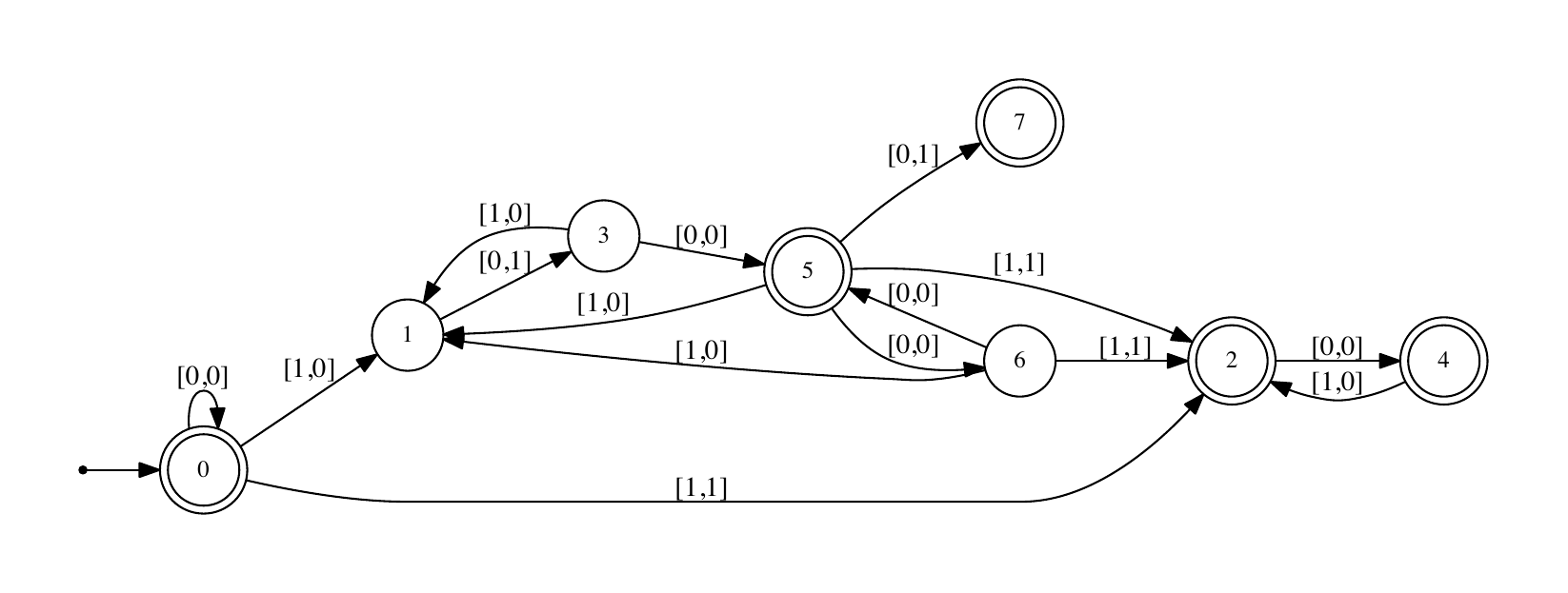}
\end{center}
\caption{Fibonacci automaton for $h'(n)$.}
\label{fig2}
\end{figure}

These automata operate as follows.  The automaton {\tt fp} 
(respectively, {\tt hp}) computes $f'(n)$ (respectively, $h'(n)$),
as follows:  one feeds the automaton with the Zeckendorf representation of
$n$ and $x$ in parallel, starting
with the most significant digits, and padding a shorter input with
leading zeros, if necessary.  Starting in state $0$, and following
the arrows, one arrives at an accepting state (depicted with a double
circle) if and only if $x = f'(n)$ (respectively, $x = h'(n)$).
Automata like these are called
``Fibonacci-synchronized'', and their properties are
discussed in \cite{Shallit:2021h}.

Technically speaking, these automata compute relations on $\Enn
\times \Enn$.
Our first step is to check that these {\it relations\/} are actually
{\it functions}.   We can do this with the following {\tt Walnut} code.
It checks that there is a value associated with every argument, and no
argument has two or more values associated with it.
\begin{verbatim}
eval func_f_check1 "?msd_fib An Ex $fp(n,x)":
eval func_f_check2 "?msd_fib An ~Ex,y x!=y & $fp(n,x) & $fp(n,y)":
eval func_h_check1 "?msd_fib An Ex $hp(n,x)":
eval func_h_check2 "?msd_fib An ~Ex,y x!=y & $hp(n,x) & $hp(n,y)":
\end{verbatim}
{\tt Walnut} returns {\tt TRUE} for each command.

It may be helpful to explain the syntax here.   {\tt eval} is
a command to evaluate the first-order formula in quotes whose
name is given immediately after the command.   The {\tt ?msd\_fib}
tells {\tt Walnut} that numbers are to be represented in
the Zeckendorf numeration system.  {\tt A} is the universal quantifier
$\forall$ and {\tt E} is the existential quantifier $\exists$.
The symbol {\tt \&} represents logical AND, and
{\tt \char'176} represents logical NOT.

Now we know that $f'$ and $h'$ are indeed natural-number valued
functions.  We now have to check that $f' = f$ and $h' = h$.
The first step is to check that
our guessed functions $f', h'$ satisfy the relation
\begin{equation}
(n+1) h'(n+1) - n h'(n) = f'(n+1)
\label{hf2}
\end{equation}
for all $n \geq 0$.  And to check
this, we will compute a linear representation for the difference
$(n+1) h'(n+1) - nh'(n) - f'(n+1)$ and verify that it 
represents the $0$ function.

Here, by a linear representation, we mean a triple
$(v, \mu, w)$, where $v$ is a $1 \times n$ vector,
$w$ is an $n \times 1$ vector, and $\mu$ is an $n\times n$ matrix-valued
morphism on the alphabet $\{0, 1\}$.  The number $n$ is called the
{\it rank\/} of the linear representation.  In this paper,
a linear representation
for a sequence $a(n)$ means that if $n$ is written as a binary
string $e_1 e_2 \cdots e_t$ in Zeckendorf representation, then
$a(n) = v \mu(e_1) \cdots \mu(e_t) w$.  For more information
about linear representations, see
\cite{Berstel&Reutenauer:2011}.

We start by using {\tt Walnut}
to construct the linear representations for $n$, $n+1$,
$h'(n)$, $h'(n+1)$, and $f'(n+1)$:
\begin{verbatim}
eval en n "?msd_fib i<n":
eval enp1 n "?msd_fib i<=n":
eval hn n "?msd_fib Ex $hp(n,x) & i<x":
eval hnp1 n "?msd_fib Ex $hp(n+1,x) & i<x":
eval fnp1 n "?msd_fib Ex $fp(n+1,x) & i<x":
\end{verbatim}
These representations have rank $6,6,8,9,$ and $19$, respectively.

From these representations we can, using
the algorithms in \cite{Berstel&Reutenauer:2011},
construct the linear representation
for $(n+1) f'(n+1) - nf'(n) - h'(n+1)$.  It has rank
121.   When we minimize this linear representation, we
get the $0$ representation.    Thus Eq.~\eqref{hf2} is proved.

Now, substituting successively $1$, $2$, $\ldots$, $n-1$ for $n$
in Eq.~\eqref{hf2}, and adding up the resulting equations,
gives
\begin{equation}
nh'(n) = f'(1) + \cdots + f'(n).
\label{hf3}
\end{equation}
In particular, since by its definition as an automaton,
the quantity $h'(n)$ is always an integer, this proves
that $\sum_{1 \leq i \leq n} f'(i)$ is divisible by $n$.
Furthermore, once we verify that $f = f'$, then Eq.~\eqref{hf3}
will show that $h = h'$.

Now it remains to verify that $f'(n) = f(n)$ for all $n$.
We do this by induction.  The base case is $n = 0$.  For
the induction step, we assume we have shown $f(i) = f'(i)$
for $0 \leq i < n$, and we want to prove it for $n$.

First, let us rule out the possibility that $f(n) > f'(n)$.
To do this, it
suffices to check that $f'(n) \not\in \{ f'(0), f'(1), \ldots, f'(n-1) \}$,
which by induction shows that
$f'(n) \not\in \{ f(0), f(1), \ldots, f(n-1) \}$.
\begin{verbatim}
eval check_fp_membership "?msd_fib An,x,i,y ($fp(n,x) & $fp(i,y) & i<n)
   => x!=y":
\end{verbatim}
and {\tt Walnut} returns {\tt TRUE}.
% eval check_f_membership "?msd_fib An,x,i,y ($f(n,x) & $f(i,y) & i<n) => x!=y":

Next, let us rule out the possibility that $f(n) < f'(n)$.  
To do this, we first prove that $f'(n) < 2n$ for all $n \geq 1$:
\begin{verbatim}
eval check_fp_inequality1 "?msd_fib An,x (n>=2 & $fp(n,x)) => x<2*n":
\end{verbatim}
% eval check_f_inequality1 "?msd_fib An,x (n>=2 & $f(n,x)) => x<2*n":
Thus the only two possibilities left are $f(n) = f'(n)$ or
$f(n) = f'(n) - n$.
To rule out the second one, it suffices to show that either
$$f'(n) - n \leq 0 \quad \text{ or } \quad f'(n) - n \in \{ f'(0), f'(1), \ldots, f'(n-1) \}.$$
\begin{verbatim}
eval check_fp_inequality2 "?msd_fib An,x $fp(n,x) => 
   (x<=n | Ei,y i<n & $fp(i,y) & x=n+y)":
\end{verbatim}
% eval check_f_inequality2 "?msd_fib An,x $f(n,x) => (x<=n | Ei,y i<n & $f(i,y) & x=n+y)":
And {\tt Walnut} returns {\tt TRUE}.
Thus we have proved that $f'(n) = f(n)$, and hence
$h'(n) = h(n)$.  From now on, then, we can
replace $f'$ with $f$ and $h'$ with $h$.

\begin{remark}
In the last part of the proof, we have also shown an alternative
characterization of the sequence $(f(n))_{n\geq 1}$; namely, that it is the
lexicographically least sequence of distinct positive integers with the
property that all values of $f(n)-n$ are also distinct.
This was observed by Ivan Neretin in the comments to
sequence \seqnum{A019444}.
\end{remark}

\section{Results of Venkatachala}
\label{sec3}

Now that we have proved the correctness of the automata,
we can easily re-prove the results of Venkatachala.  All we have to do
is translate his claims into first-order logic.
\begin{theorem}
\leavevmode
\begin{itemize}
\item[(a)] $h$ is a nondecreasing function of $n$ .
\item[(b)] $h(n) \leq n$ for $n \geq 0$.
\item[(c)] $h(n+1) \in \{ h(n), h(n)+1 \}$ for $n \geq 0$.
\item[(d)] If $n \geq 0$ then $h(n+1)=h(n) \Leftrightarrow f(n+1) = h(n)$.
\item[(e)] If $n \geq 0$ then $h(n+1)=h(n)+1 \Leftrightarrow f(n+1) = h(n)+n+1$.
\item[(f)] If $n \geq 1$ then $h(h(n))+h(n+1) = n+2$.
\item[(g)] If $n \geq 0$ then $f(f(n)) = n$.
\item[(h)] If $n \geq 1$ then $h(h(n)+n) = n+1$.
\item[(i)] If $n \geq 1$ then $h(n) = \lfloor n \varphi \rfloor -n+1$.
\item[(j)] The function $h$ does not assume the same value at three
distinct integers.
\item[(k)] For $n\geq 0$ we have $h(n+2) > h(n)$.
\item[(l)] The function $h$ is onto.
\item[(m)] The function $f$ is onto.
\item[(n)] The function $f$ is one-to-one.
\item[(o)] For $n \geq 6$ we have $h(n) \leq n-2$.
\item[(p)] If $f(n+1)> h(n)$, then $f(j)>h(n)$ for all $j\geq n+1$.
\item[(q)] There are no integers $k \geq 2$ and $\ell$ such that
$f(k-1) = \ell$ and $f(k) = \ell+1$.
\end{itemize}
\end{theorem}

\begin{proof}
We will need {\tt Walnut} code for $\lfloor n \varphi \rfloor$,
$\lfloor n/\varphi \rfloor$, and
$\lfloor n \varphi^2 \rfloor$.  These can be found
in \cite{Shallit:2022}.
\begin{verbatim}
reg shift {0,1} {0,1} "([0,0]|[0,1][1,1]*[1,0])*":
def phin "?msd_fib (s=0 & n=0) | Ex $shift(n-1,x) & s=x+1":
def phi2n "?msd_fib (s=0 & n=0) | Ex,y $shift(n-1,x) &
   $shift(x,y) & s=y+2":
def noverphi "?msd_fib Et $phin(n,t) & s+n=t":
\end{verbatim}

We use the following {\tt Walnut} code.
\begin{verbatim}
eval testa "?msd_fib An,x,y ($h(n,x) & $h(n+1,y)) => x<=y":
eval testb "?msd_fib An,x $h(n,x) => x<=n":
eval testc "?msd_fib An,x,y ($h(n,x) & $h(n+1,y)) => (y=x | y=x+1)":
eval testd "?msd_fib An,x,y,z ($h(n,x) & $h(n+1,y) & $f(n+1,z)) => 
   (y=x <=> z=x)":
eval teste "?msd_fib An,x,y,z ($h(n,x) & $h(n+1,y) & $f(n+1,z)) => 
   (y=x+1 <=> z=x+n+1)":
eval testf "?msd_fib An,x,y,z (n>=1 & $h(n,x) & $h(x,y) 
   & $h(n+1,z)) => y+z=n+2":
eval testg "?msd_fib An,x,y ($f(n,x) & $f(x,y)) => y=n":
eval testh "?msd_fib An,x,y (n>=1 & $h(n,x) & $h(x+n,y)) 
   => y=n+1":
eval testi "?msd_fib An,x,y (n>=1 & $h(n,x) & $phin(n,y)) 
   => x+n=y+1":
eval testj "?msd_fib ~En1,n2,n3,x (n1<n2) & (n2<n3) & 
   $h(n1,x) & $h(n2,x) & $h(n3,x)":
eval testk "?msd_fib An,x,y ($h(n,x) & $h(n+2,y)) => y>x":
eval testl "?msd_fib Ax En $h(n,x)":
eval testm "?msd_fib Ax En $f(n,x)":
eval testn "?msd_fib An1,n2 (Ex $f(n1,x) & $f(n2,x)) => n1=n2":
eval testo "?msd_fib An,x (n>=6 & $h(n,x)) => x+2<=n":
eval testp "?msd_fib An,x,y ($f(n+1,x) & $h(n,y) & x>y) => 
   (Aj,t (j>=n+1 & $f(j,t)) => t>y)":
eval testq "?msd_fib ~Ek,l k>=2 & $f(k-1,l) & $f(k,l+1)":
\end{verbatim}
And all of them return {\tt TRUE}.
\end{proof}

For a different approach, see the recent paper of Dekking
\cite{Dekking:2023}.

\section{Automata and the Avdispahi\'c \& Zejnulahi sequences}

Once again, we can easily guess candidate Fibonacci automata for these
sequences from their initial values.   Once guessed, we can verify
their correctness exactly as we did in Sections~\ref{sec2}.
We omit the details.

The automaton for $z$ has $18$ states and that for $m$ has $8$ states.
They are depicted in Figures~\ref{fig3} and \ref{fig4}.
\begin{figure}[htb]
\begin{center}
\includegraphics[width=6in]{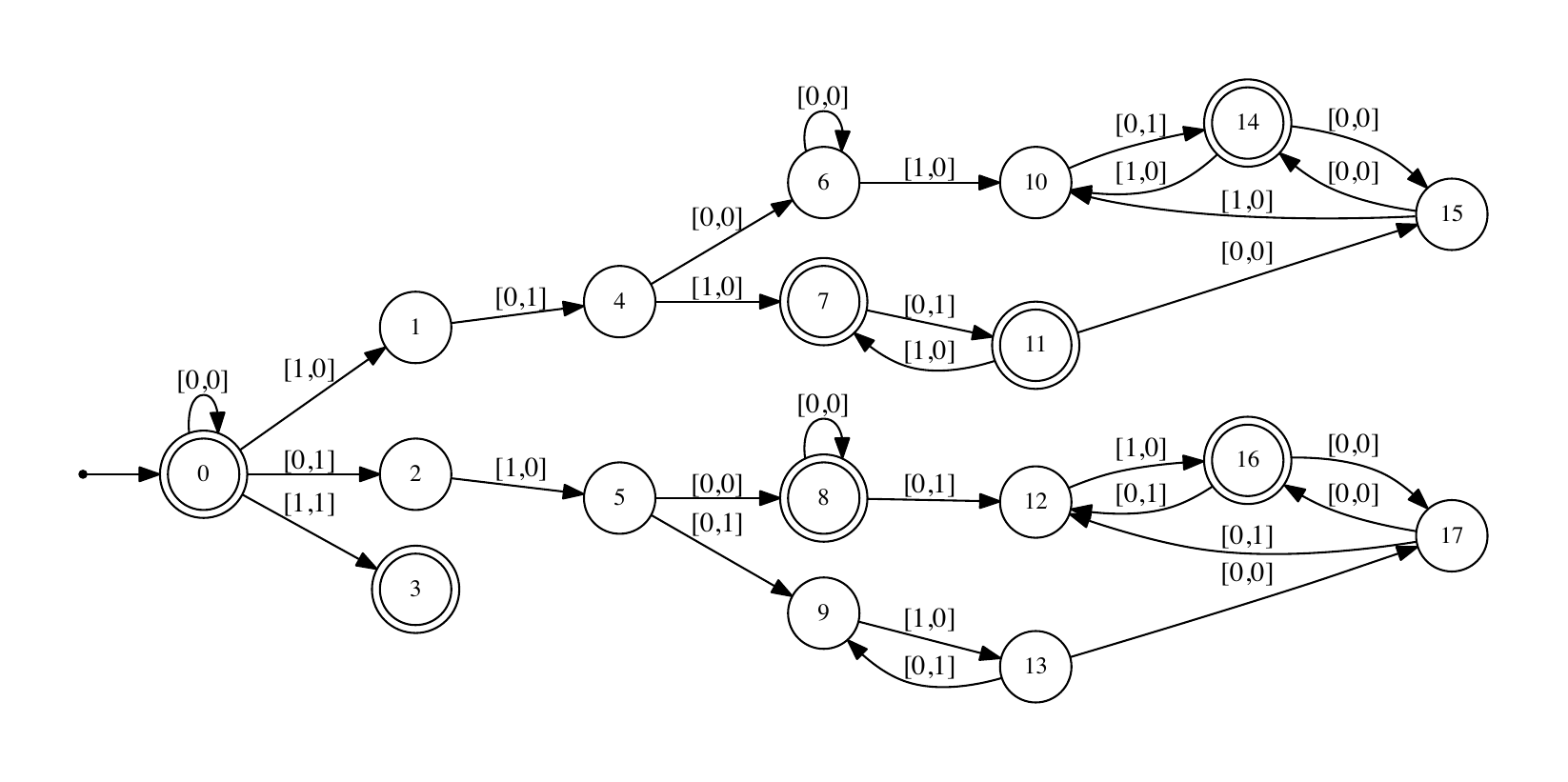}
\end{center}
\caption{Fibonacci automaton for $z(n)$.}
\label{fig3}
\end{figure}

\begin{figure}[htb]
\begin{center}
\includegraphics[width=6in]{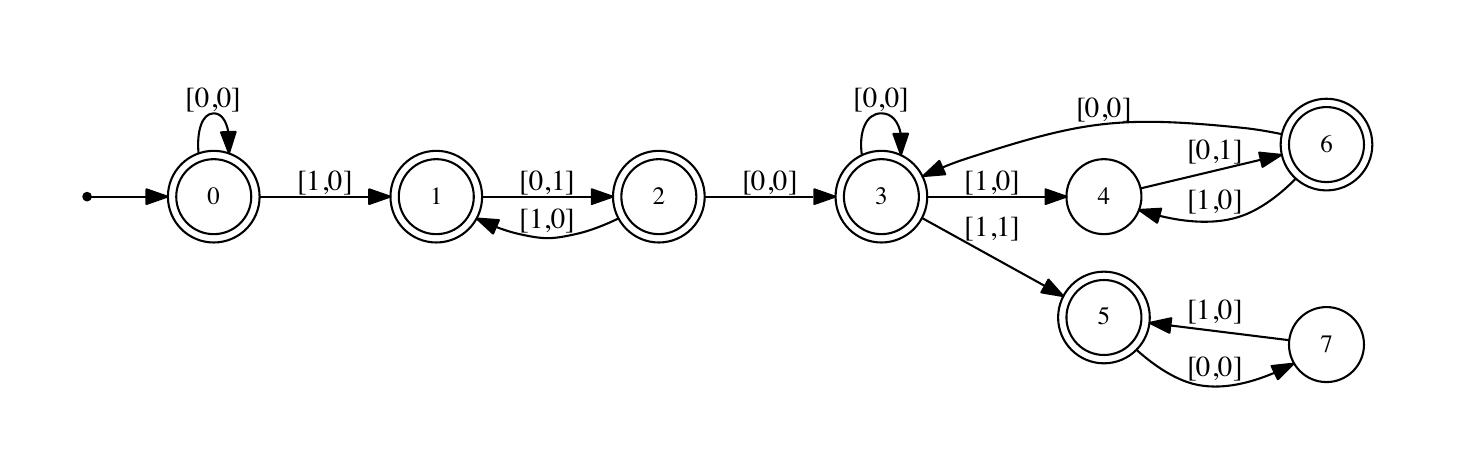}
\end{center}
\caption{Fibonacci automaton for $m(n)$.}
\label{fig4}
\end{figure}

Once we have the automata we can easily prove the following
results of Avdispahi\'c \& Zejnulahi:
\begin{theorem}
\leavevmode
\begin{itemize}
\item[(a)] If $n = F_k$ for $k \geq 2$ then $z(n) = F_{k+1}$.
\item[(b)] If $n = F_k - 1$ for $k \geq 4$ then $z(n) = F_{k-1} - 1$.
\item[(c)] 
If $n = \lfloor k\tau^2 \rfloor$ and
$n \not= F_i$ for $i>2$ and $n \not= F_i - 1$ for $i>4$, then
$z(n) = \lfloor k\tau \rfloor$.
\item[(d)]
If $n = \lfloor k\tau \rfloor$ and
$n \not= F_i$ for $i>2$ and $n \not= F_i - 1$ for $i>4$, then
$z(n) = \lfloor k\tau^2 \rfloor$.
\item[(e)] If $n>2$ and $m(n-1) \not= z(i)$ for $1 \leq i < n$,
then $z(n) = m(n)= m(n-1)$.
\item[(f)] If $n>2$ and $m(n-1) = z(i)$ for some $i$, $1 \leq i < n$,
then $z(n) = m(n-1)+n+1$ and $m(n)=m(n-1)+1$.
\item[(g)] The sequence $z(n)$ is onto.
\item[(h)] The sequence $z(n)$ is one-one.
\end{itemize}
\end{theorem}

\begin{proof}
We use the following {\tt Walnut} code:
\begin{verbatim}
reg isfib msd_fib "0*10*":
reg adjfib msd_fib msd_fib "[0,0]*[0,1][1,0][0,0]*":
eval parta "?msd_fib Ax,y ($adjfib(x,y) & x>=2) => $zp(x,y)":
eval partb "?msd_fib Ax,y ($adjfib(x,y) & y>=5) => $zp(y-1,x-1)":
eval partc "?msd_fib Ak,n,x ($phi2n(k,n) & $phin(k,x) & (~$isfib(n)) 
   & (~$isfib(n+1))) => $zp(n,x)":
eval partd "?msd_fib Ak,n,x ($phin(k,n) & $phi2n(k,x) & (~$isfib(n)) &
   (~$isfib(n+1))) => $zp(n,x)":
eval parte "?msd_fib An,x,y,w (n>2 & $zp(n,x) & $mp(n,y) & $mp(n-1,w) & 
(Ai,r (i>=1 & i<n & $z(i,r)) => r!=w)) => (x=w & y=w)":
eval partf "?msd_fib An,x,y,w (n>2 & $zp(n,x) & $mp(n,y) & $mp(n-1,w) &
(Ei,r i>=1 & i<n & $zp(i,r) & r=w)) => (x=w+n+1 & y=w+1)":
eval partg "?msd_fib Ax En $zp(n,x)":
eval parth "?msd_fib An1,n2 (Ex $zp(n1,x) & $zp(n2,x)) => n1=n2":
\end{verbatim}
All of the {\tt Walnut} commands return {\tt TRUE}.
\end{proof}

\section{New results}

The advantage to our approach is that, once the automata are obtained,
it becomes almost trivial to test additional conjectures and prove
new results.  We give a few examples.

\begin{theorem}
Suppose $n \geq 1$.  Then
$$f(n) = \begin{cases}
	\lfloor n \varphi \rfloor, & \text{ if $\exists m \ n-1 = \lfloor m \varphi \rfloor $;} \\
	\lfloor n/\varphi \rfloor + 1, & \text{otherwise.}
	\end{cases}
$$
\end{theorem}

\begin{proof}
We use the following {\tt Walnut} code:
\begin{verbatim}
eval chk1 "?msd_fib An,x (n>=1 & $fp(n,x) & Em $phin(m,n-1)) 
   => $phin(n,x)":
eval chk2 "?msd_fib An,x (n>=1 & $fp(n,x) & ~Em $phin(m,n-1)) 
   => $noverphi(n,x-1)":
\end{verbatim}
\end{proof}

\begin{corollary}
$\lfloor n/\varphi \rfloor + 1 \leq 
f(n) \leq \lfloor n \varphi \rfloor$ for all $n \geq 1$.
\end{corollary}

We now prove a theorem linking the functions of the two papers.

\begin{theorem}
We have $z(n) \in \{ f(n), f(n)+n, f(n)+1, f(n)-n \}$ for all $n \geq 0$.
\end{theorem}

\begin{proof}
We use the following {\tt Walnut} code:
\begin{verbatim}
eval thm6 "?msd_fib An,x,y ($fp(n,x) & $zp(n,y)) => (y=x|y=x+n|y=x+1|y+n=x)":
\end{verbatim}
and {\tt Walnut} returns {\tt TRUE}.
\end{proof}

\begin{remark}
One may reasonably ask for simple characterizations of the $n$ for which
each case occurs.   The answer is that $n$ for which
$z(n) = f(n)$, $z(n)=f(n)+n$, $z(n)=f(n)+1$, $z(n)=f(n)-n$
are accepted by automata of $7$, $5$, $6$, and $9$
states, respectively, and these are easy to compute with
{\tt Walnut}.
% eval test1 "?msd_fib Ex,y $fp(n,x) & $zp(n,y) & y=x":
% eval test2 "?msd_fib Ex,y $fp(n,x) & $zp(n,y) & y=x+n":
% eval test3 "?msd_fib Ex,y $fp(n,x) & $zp(n,y) & y=x+1":
% eval test4 "?msd_fib Ex,y $fp(n,x) & $zp(n,y) & y+n=x":

Furthermore, by examining these automata, we easily see that
$z(n) = f(n)$ if and only if $n= 0$ or $n= F_{2k+1}$ for $k \geq 0$.

Similarly,
$z(n)=f(n)+n$ if and only if $n=0$ or $n = \lfloor k \varphi^2 \rfloor + 1$
for some $k \geq 1$.  This can easily be proved with {\tt Walnut}.
\end{remark}

\section{Going further}

\subsection{Hofstadter's married functions}

Hofstadter \cite{Hofstadter:1979} defined two sequences he called
``married'' functions, via the system of recurrences
\begin{align*}
b(n) &= n - a(b(n-1)) \\
a(n) &= n-b (a(n-1)),
\end{align*}
for $n \geq 1$, with initial values $a(0) = 1$, $b(0) = 0$.
Here $(a(n))$ is \seqnum{A005378} and $(b(n))$ is \seqnum{A005379}.
The first few values are given in Table~\ref{tab10}.
\begin{table}[H]
\begin{center}
\begin{tabular}{c|ccccccccccccccccccccc}
$n$ & 0 & 1 & 2 & 3 & 4 & 5 & 6 & 7 & 8 & 9 & 10 & 11 & 12 & 13 & 14 & 15 & 16 & 17 & 18   \\
\hline
$a(n)$ & 1& 1& 2& 2& 3& 3& 4& 5& 5& 6& 6& 7& 8& 8& 9& 9&10&11&11 \\
$b(n)$ & 0& 0& 1& 2& 2& 3& 4& 4& 5& 6& 6& 7& 7& 8& 9& 9&10&11&11
\end{tabular}
\end{center}
\caption{Hofstadter's ``married'' functions.}
\label{tab10}
\end{table}

We can use our techniques to guess and prove
the automata for these functions.  They are depicted below
in Figure~\ref{fig5}.
\begin{figure}[H]
\begin{center}
\includegraphics[width=3.2in]{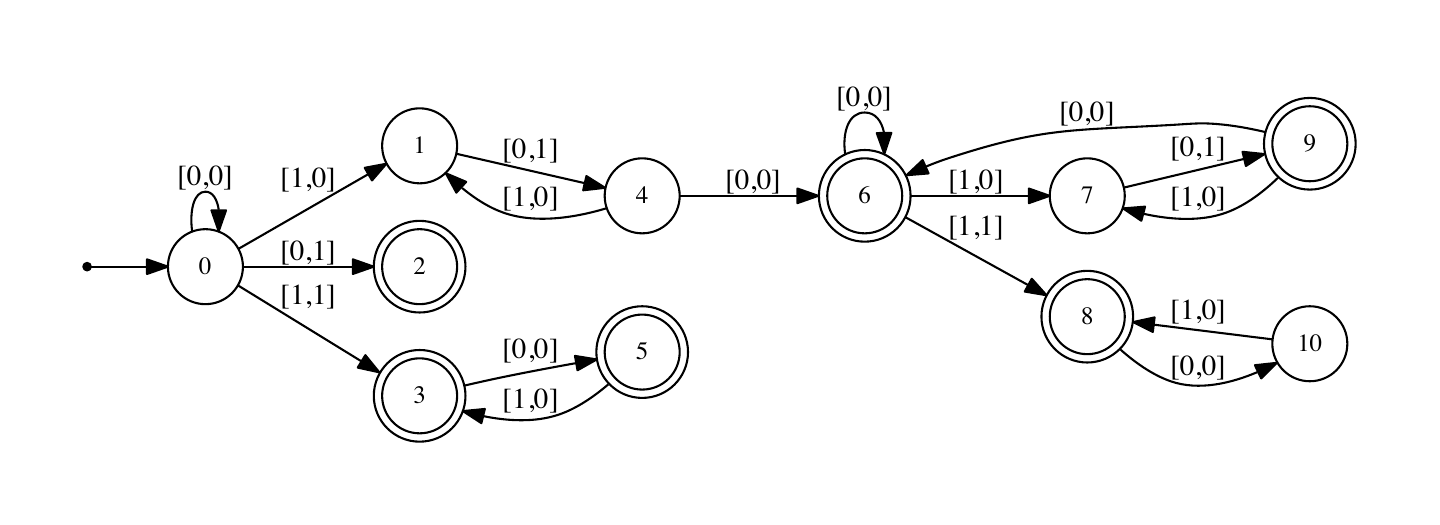} 
\includegraphics[width=3.2in]{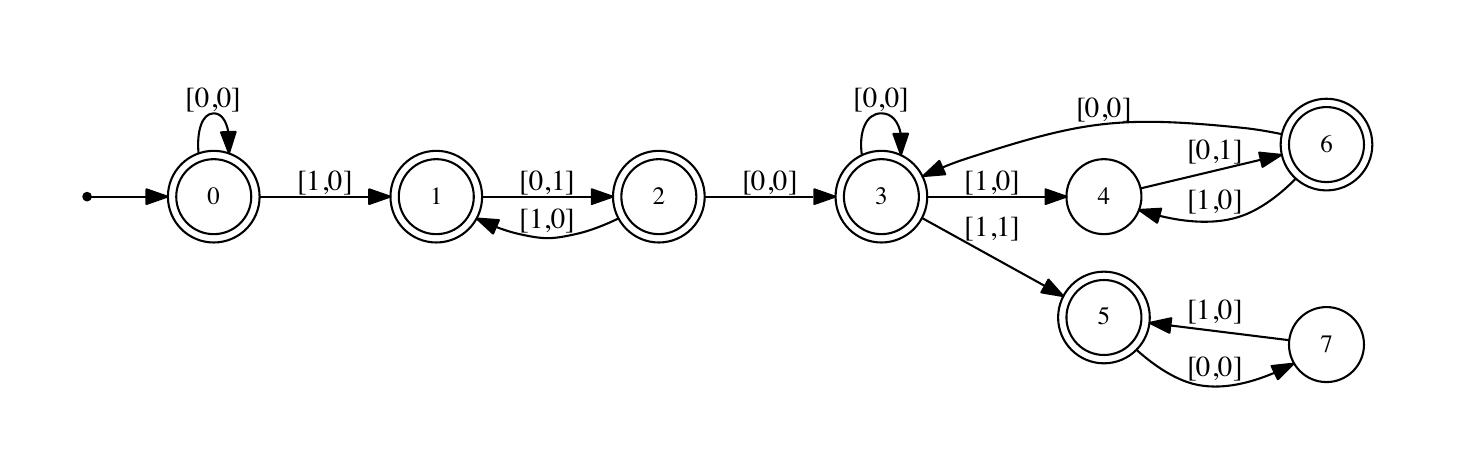}
\end{center}
\caption{Fibonacci automata {\tt ha} and {\tt hb} for the ``married'' functions $a(n)$ 
(left) and $b(n)$ (right).}
\label{fig5}
\end{figure}

With these automata, we can easily recover the closed forms for
these sequences previously obtained by Stoll \cite{Stoll:2008}:
\begin{theorem}
Let $\alpha = (\sqrt{5}-1)/2$.  Define
$$\varepsilon_1 (n) = \begin{cases}
	1, & \text{if $n= F_{2k} -1$ for some $k \geq 1$;} \\
	0, & \text{otherwise;}
	\end{cases}
$$
and 
$$\varepsilon_2 (n) = \begin{cases}
        1, & \text{if $n= F_{2k+1} -1$ for some $k \geq 1$;}\\
        0, & \text{otherwise;}
        \end{cases}
$$
Then for $n \geq 0$ we have
\begin{align*}
a(n) &= \lfloor (n+1) \alpha \rfloor + \epsilon_1 (n) \\
b(n) &= \lfloor (n+1) \alpha \rfloor - \epsilon_2 (n)
\end{align*}
\end{theorem}

\begin{proof}
We use the following {\tt Walnut} code:
\begin{verbatim}
reg evenfib msd_fib "0*1(00)*":
reg oddfib msd_fib "0*10(00)*":

def eps1 "?msd_fib (x=1 & $evenfib(n+1)) | (x=0 & ~$evenfib(n+1))":
def eps2 "?msd_fib (x=1 & $oddfib(n+1)) | (x=0 & ~$oddfib(n+1))":

eval checkstolla "?msd_fib An,x,y,z ($noverphi(n+1,x) & $ha(n,y) &
   $eps1(n,z)) => y=x+z":
eval checkstollb "?msd_fib An,x,y,z ($noverphi(n+1,x) & $hb(n,y) &
   $eps2(n,z)) => y+z=x":
\end{verbatim}
And {\tt TRUE} is returned twice.
\end{proof}

\subsection{Quet's sequence}

In a personal communication, Muharem Avdispahi\'c suggested
looking at parameterized versions of the sequences
we have studied here.  Here is one possibility:  let $k \geq -1$. 
Define $A_k (n) = n$ for $0 \leq n \leq k$,
and for $n \geq k$ define $A_k(n+1)$ to be the least natural number
such that $A_k(n+1)  \not\in \{ A_k(0), A_k(1), \ldots, A_k(n)\}$ and
$\sum_{k+1 \leq i \leq n+1} A_k(n) \equiv \modd{0} {n+k}$.  
Values of the first few sequences are given in
Table~\ref{tab6}.

\begin{table}[H]
\begin{center}
\begin{tabular}{c|ccccccccccccccccccccc}
$n$ & 0 & 1 & 2 & 3 & 4 & 5 & 6 & 7 & 8 & 9 & 10 & 11 & 12 & 13 & 14 & 15 & 16 & 17 & 18   \\
\hline
$A_{-1} (n)$ & 0& 1& 2& 3& 6& 4& 9& 5&12&14& 7&17& 8&20&22&10&25&11&28\\
$A_0(n)$ &  0& 1& 3& 2& 6& 8& 4&11& 5&14&16& 7&19&21& 9&24&10&27&29\\
$A_1(n)$ &  0& 1& 3& 5& 2& 8&10& 4&13&15& 6&18& 7&21&23& 9&26&28&11\\
$A_2(n)$ &  0& 1& 2& 5& 7& 9& 3&12& 4&15&17& 6&20&22& 8&25&27&10&30\\
$A_3(n)$ &  0& 1& 2& 3& 7& 9&11&13& 4&16& 5&19& 6&22&24& 8&27&29&10\\
$A_4(n)$ &  0& 1& 2& 3& 4& 9&11&13&15&17& 5&20& 6&23& 7&26& 8&29&31
\end{tabular}
\end{center}
\caption{Parameterized sequences.}
\label{tab6}
\end{table}

Note that $A_0(n) = f(n)$ and $A_1(n) = z(n)$.

Of particular interest is $A_{-1} (n)$.  This
is Quet's sequence $(\seqnum{A125147}(n-1))_{n \geq 0}$.
As an example of the usefulness of the automaton-based method,
we can now prove a conjecture about Quet's sequence.

\begin{theorem}
Quet's sequence is a permutation of the integers.
\end{theorem}

\begin{proof}
We follow the same sequence of steps as
in Section~\ref{sec2} and \ref{sec3}:  first, define
\begin{equation}
B_{-1} (n+1) = {1 \over n} \left(\sum_{0 \leq i \leq n} A_{-1}(i) \right)
\label{quetb}
\end{equation}
for $n \geq 1$, and set $B_{-1}(i) = i$ for $i = 0,1$.
The first few values of $B_{-1}(n)$ are given in 
Table~\ref{tab7}.
\begin{table}[H]
\begin{center}
\begin{tabular}{c|ccccccccccccccccccccc}
$n$ & 0 & 1 & 2 & 3 & 4 & 5 & 6 & 7 & 8 & 9 & 10 & 11 & 12 & 13 & 14 & 15 & 16 & 17 & 18   \\
\hline
$B_{-1} (n)$ & 0& 1& 3& 3& 4& 4& 5& 5& 6& 7& 7& 8& 8& 9&10&10&11&11&12
\end{tabular}
\end{center}
\caption{Values of $B_{-1} (n)$.}
\label{tab7}
\end{table}

Next, from the first
few computed values of $A_{-1}(n)$ and $B_{-1} (n)$,
we ``guess'' an automaton
for $A_{-1}(n)$ and its associated sequence $B_{-1}(n)$ 
that we hope will verify Eq.~\eqref{quetb}.

These automata have 29 and 17 states, respectively, and are depicted
in Figures~\ref{fig6} and \ref{fig7}.
\begin{figure}[H]
\begin{center}
\includegraphics[width=6in]{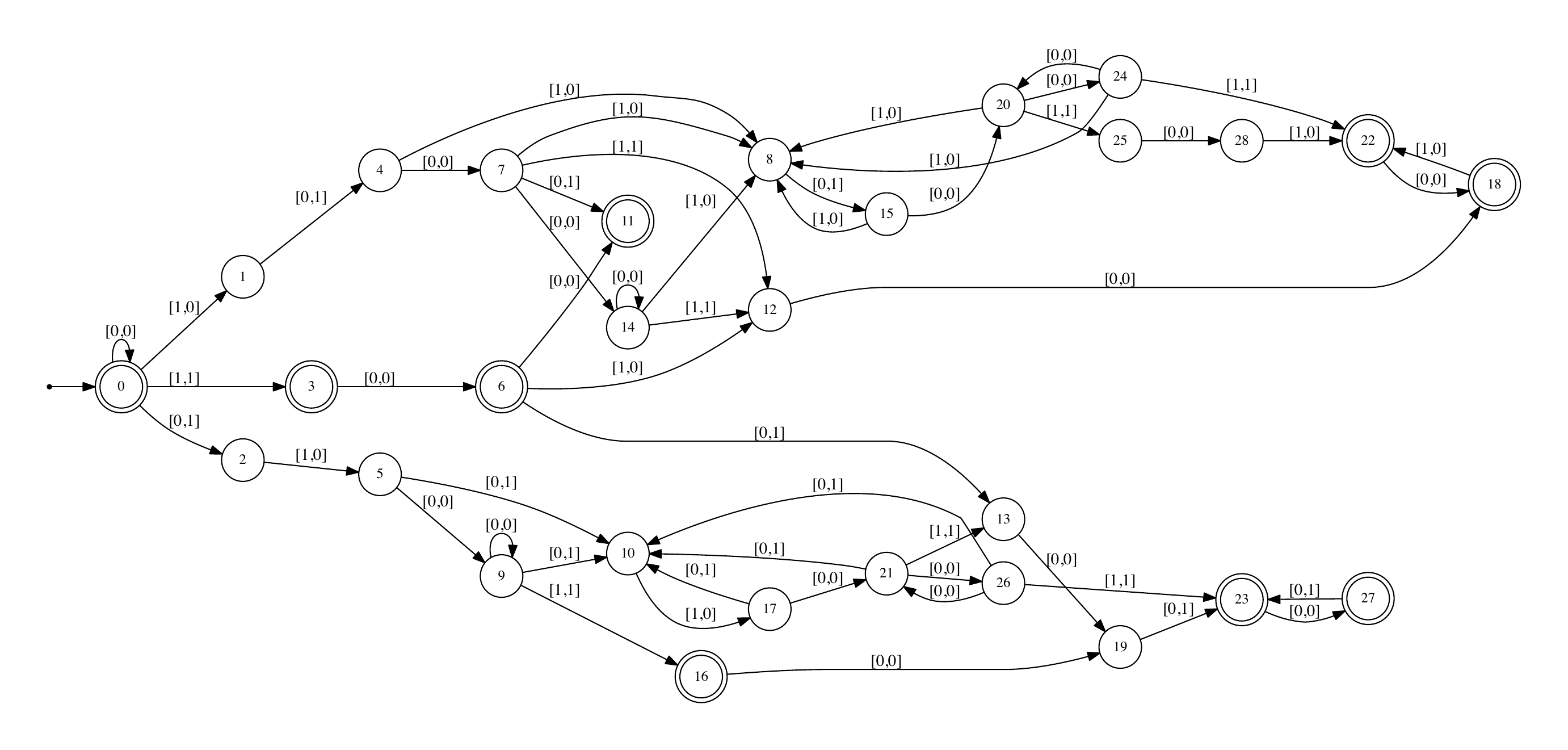} \
\end{center}
\caption{Fibonacci automaton {\tt queta} for $A_{-1}(n)$.}
\label{fig6}
\end{figure}

\begin{figure}[H]
\begin{center}
\includegraphics[width=6in]{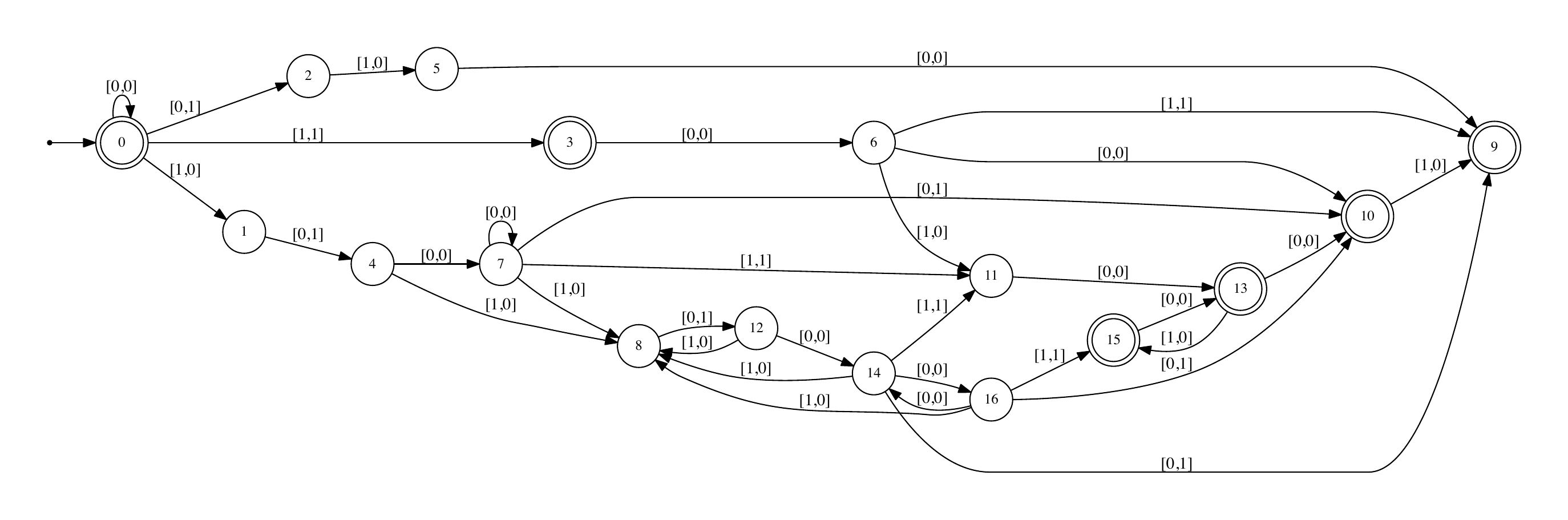} \
\end{center}
\caption{Fibonacci automaton {\tt quetb} for $B_{-1}(n)$.}
\label{fig7}
\end{figure}

Next, we verify that the guessed automata actually compute
functions, say $A'_{-1} (n)$ and $B'_{-1}(n)$.

We then check the correctness of the initial values and,
using the linear representations of the automata,
check that
$$ A'_{-1} (n+3) = (n+2) B'_{-1} (n+3) - (n+1) B'_{-1} (n+2)$$
for all $n \geq 0$.
This proves that the automata really do compute the
functions $A_{-1}(n)$ and $B_{-1} (n)$.

Finally, we use {\tt Walnut} to check the assertion
that $A_{-1} (n)$ is a permutation of $\Enn$.

The needed {\tt Walnut} code is below:
\begin{verbatim}
eval queta_check1 "?msd_fib An Ex $queta(n,x)":
eval queta_check2 "?msd_fib An ~Ex,y x!=y & $queta(n,x) & $queta(n,y)":
eval quetb_check1 "?msd_fib An Ex $quetb(n,x)":
eval quetb_check2 "?msd_fib An ~Ex,y x!=y & $quetb(n,x) & $quetb(n,y)":

eval en n "?msd_fib i<n":
eval enp1 n "?msd_fib i<n+1":
eval enp2 n "?msd_fib i<n+2":
#6,6,8 states

eval qan3 n "?msd_fib Ex $queta(n+3,x) & i<x":
eval qbn3 n "?msd_fib Ex $quetb(n+3,x) & i<x":
eval qbn2 n "?msd_fib Ex $quetb(n+2,x) & i<x":
#52,20,17 states

#check that queta is a permutation of the integers
eval queta_onto "?msd_fib Ax En $queta(n,x)":
eval queta_one_one "?msd_fib ~En1,n2,x (n1!=n2) & $queta(n1,x) & $queta(n2,x)":
\end{verbatim}
\end{proof}

\subsection{An unexpected connection}

Finally, we show that Quet's sequence is intimately connected
to the ``married'' functions of Hofstadter:
\begin{theorem}
For all $n \geq 2$ we have $B_{-1} (n) = a(n-2)+2$.
\end{theorem}

\begin{proof}
We use the following {\tt Walnut} command:
\begin{verbatim}
eval quetbcheck "?msd_fib An,x,y (n>=2 & $quetb(n,x) & $ha(n-2,y)) => x=y+2":
\end{verbatim}
which returns {\tt TRUE}.
\end{proof}

{\tt Walnut} is available at the following URL:\\
\centerline{\url{https://cs.uwaterloo.ca/~shallit/walnut.html} \ .}

\section{Acknowledgment}

I thank Muharem Avdispahi\'c for his kind suggestion, and Michel Dekking for
pleasant and helpful discussions.

\end{document}